\documentclass[a4paper,UKenglish]{lipics-v2018}
 \nolinenumbers
\usepackage{microtype}
\usepackage[ruled]{algorithm2e} 
\usepackage{tikz}

\newcommand{\N}{\mathbb{N}}

\newcommand{\Exp}[1]{\mathbb{E}\left[#1\right]} 

\newcommand{\ignore}[1]{}

\title{Speed Scaling with Tandem Servers}

\author{Rahul Vaze}{School of Technology and Computer Science, \\Tata Institute of Fundamental Research, Mumbai, India}{rahul.vaze@gmail.com}{}{}
\author{Jayakrishnan Nair}{Department of Electrical Engineering, \\Indian Institute of Technology, Bombay, India}{jayakrishnan.nair@ee.iitb.ac.in}{}{}
\authorrunning{R.\, Vaze and J.\, Nair}

\subjclass{F.2.0: ANALYSIS OF ALGORITHMS AND PROBLEM COMPLEXITY}
\Copyright{Rahul\, Vaze and Jayakrishnan\, Nair}
\keywords{Speed Scaling, Online Algorithms, Tandem Servers}

\newboolean{showcomments}
\setboolean{showcomments}{true}
\newcommand{\jk}[1]{  \ifthenelse{\boolean{showcomments}}
{ \textcolor{red}{(JK says:  #1)}} {}  }
\newcommand{\rv}[1]{  \ifthenelse{\boolean{showcomments}}
{ \textcolor{red}{(RV says:  #1)}} {}  }
\begin{document}
\maketitle
\def\bba{{\mathbb{a}}}
\def\bbb{{\mathbb{b}}}
\def\bbc{{\mathbb{c}}}
\def\bbd{{\mathbb{d}}}
\def\bbee{{\mathbb{e}}}
\def\bbff{{\mathbb{f}}}
\def\bbg{{\mathbb{g}}}
\def\bbh{{\mathbb{h}}}
\def\bbi{{\mathbb{i}}}
\def\bbj{{\mathbb{j}}}
\def\bbk{{\mathbb{k}}}
\def\bbl{{\mathbb{l}}}
\def\bbm{{\mathbb{m}}}
\def\bbn{{\mathbb{n}}}
\def\bbo{{\mathbb{o}}}
\def\bbp{{\mathbb{p}}}
\def\bbq{{\mathbb{q}}}
\def\bbr{{\mathbb{r}}}
\def\bbs{{\mathbb{s}}}
\def\bbt{{\mathbb{t}}}
\def\bbu{{\mathbb{u}}}
\def\bbv{{\mathbb{v}}}
\def\bbw{{\mathbb{w}}}
\def\bbx{{\mathbb{x}}}
\def\bby{{\mathbb{y}}}
\def\bbz{{\mathbb{z}}}
\def\bb0{{\mathbb{0}}}

\def\bydef{:=}
\def\ba{{\mathbf{a}}}
\def\bb{{\mathbf{b}}}
\def\bc{{\mathbf{c}}}
\def\bd{{\mathbf{d}}}
\def\bee{{\mathbf{e}}}
\def\bff{{\mathbf{f}}}
\def\bg{{\mathbf{g}}}
\def\bh{{\mathbf{h}}}
\def\bi{{\mathbf{i}}}
\def\bj{{\mathbf{j}}}
\def\bk{{\mathbf{k}}}
\def\bl{{\mathbf{l}}}
\def\bm{{\mathbf{m}}}
\def\bn{{\mathbf{n}}}
\def\bo{{\mathbf{o}}}
\def\bp{{\mathbf{p}}}
\def\bq{{\mathbf{q}}}
\def\br{{\mathbf{r}}}
\def\bs{{\mathbf{s}}}
\def\bt{{\mathbf{t}}}
\def\bu{{\mathbf{u}}}
\def\bv{{\mathbf{v}}}
\def\bw{{\mathbf{w}}}
\def\bx{{\mathbf{x}}}
\def\by{{\mathbf{y}}}
\def\bz{{\mathbf{z}}}
\def\b0{{\mathbf{0}}}
\def\alg{\mathsf{ALG}}
\def\opt{\mathsf{OPT}}
\def\off{\mathsf{OFF}}
\def\bA{{\mathbf{A}}}
\def\bB{{\mathbf{B}}}
\def\bC{{\mathbf{C}}}
\def\bD{{\mathbf{D}}}
\def\bE{{\mathbf{E}}}
\def\bF{{\mathbf{F}}}
\def\bG{{\mathbf{G}}}
\def\bH{{\mathbf{H}}}
\def\bI{{\mathbf{I}}}
\def\bJ{{\mathbf{J}}}
\def\bK{{\mathbf{K}}}
\def\bL{{\mathbf{L}}}
\def\bM{{\mathbf{M}}}
\def\bN{{\mathbf{N}}}
\def\bO{{\mathbf{O}}}
\def\bP{{\mathbf{P}}}
\def\bQ{{\mathbf{Q}}}
\def\bR{{\mathbf{R}}}
\def\bS{{\mathbf{S}}}
\def\bT{{\mathbf{T}}}
\def\bU{{\mathbf{U}}}
\def\bV{{\mathbf{V}}}
\def\bW{{\mathbf{W}}}
\def\bX{{\mathbf{X}}}
\def\bY{{\mathbf{Y}}}
\def\bZ{{\mathbf{Z}}}
\def\b1{{\mathbf{1}}}

\def\bbA{{\mathbb{A}}}
\def\bbB{{\mathbb{B}}}
\def\bbC{{\mathbb{C}}}
\def\bbD{{\mathbb{D}}}
\def\bbE{{\mathbb{E}}}
\def\bbF{{\mathbb{F}}}
\def\bbG{{\mathbb{G}}}
\def\bbH{{\mathbb{H}}}
\def\bbI{{\mathbb{I}}}
\def\bbJ{{\mathbb{J}}}
\def\bbK{{\mathbb{K}}}
\def\bbL{{\mathbb{L}}}
\def\bbM{{\mathbb{M}}}
\def\bbN{{\mathbb{N}}}
\def\bbO{{\mathbb{O}}}
\def\bbP{{\mathbb{P}}}
\def\bbQ{{\mathbb{Q}}}
\def\bbR{{\mathbb{R}}}
\def\bbS{{\mathbb{S}}}
\def\bbT{{\mathbb{T}}}
\def\bbU{{\mathbb{U}}}
\def\bbV{{\mathbb{V}}}
\def\bbW{{\mathbb{W}}}
\def\bbX{{\mathbb{X}}}
\def\bbY{{\mathbb{Y}}}
\def\bbZ{{\mathbb{Z}}}

\def\cA{\mathcal{A}}
\def\cB{\mathcal{B}}
\def\cC{\mathcal{C}}
\def\cD{\mathcal{D}}
\def\cE{\mathcal{E}}
\def\cF{\mathcal{F}}
\def\cG{\mathcal{G}}
\def\cH{\mathcal{H}}
\def\cI{\mathcal{I}}
\def\cJ{\mathcal{J}}
\def\cK{\mathcal{K}}
\def\cL{\mathcal{L}}
\def\cM{\mathcal{M}}
\def\cN{\mathcal{N}}
\def\cO{\mathcal{O}}
\def\cP{\mathcal{P}}
\def\cQ{\mathcal{Q}}
\def\cR{\mathcal{R}}
\def\cS{\mathcal{S}}
\def\cT{\mathcal{T}}
\def\cU{\mathcal{U}}
\def\cV{\mathcal{V}}
\def\cW{\mathcal{W}}
\def\cX{\mathcal{X}}
\def\cY{\mathcal{Y}}
\def\cZ{\mathcal{Z}}

\def\sfA{\mathsf{A}}
\def\sfB{\mathsf{B}}
\def\sfC{\mathsf{C}}
\def\sfD{\mathsf{D}}
\def\sfE{\mathsf{E}}
\def\sfF{\mathsf{F}}
\def\sfG{\mathsf{G}}
\def\sfH{\mathsf{H}}
\def\sfI{\mathsf{I}}
\def\sfJ{\mathsf{J}}
\def\sfK{\mathsf{K}}
\def\sfL{\mathsf{L}}
\def\sfM{\mathsf{M}}
\def\sfN{\mathsf{N}}
\def\sfO{\mathsf{O}}
\def\sfP{\mathsf{P}}
\def\sfQ{\mathsf{Q}}
\def\sfR{\mathsf{R}}
\def\sfS{\mathsf{S}}
\def\sfT{\mathsf{T}}
\def\sfU{\mathsf{U}}
\def\sfV{\mathsf{V}}
\def\sfW{\mathsf{W}}
\def\sfX{\mathsf{X}}
\def\sfY{\mathsf{Y}}
\def\sfZ{\mathsf{Z}}

\def\bydef{:=}
\def\sfa{{\mathsf{a}}}
\def\sfb{{\mathsf{b}}}
\def\sfc{{\mathsf{c}}}
\def\sfd{{\mathsf{d}}}
\def\sfee{{\mathsf{e}}}
\def\sfff{{\mathsf{f}}}
\def\sfg{{\mathsf{g}}}
\def\sfh{{\mathsf{h}}}
\def\sfi{{\mathsf{i}}}
\def\sfj{{\mathsf{j}}}
\def\sfk{{\mathsf{k}}}
\def\sfl{{\mathsf{l}}}
\def\sfm{{\mathsf{m}}}
\def\sfn{{\mathsf{n}}}
\def\sfo{{\mathsf{o}}}
\def\sfp{{\mathsf{p}}}
\def\sfq{{\mathsf{q}}}
\def\sfr{{\mathsf{r}}}
\def\sfs{{\mathsf{s}}}
\def\sft{{\mathsf{t}}}
\def\sfu{{\mathsf{u}}}
\def\sfv{{\mathsf{v}}}
\def\sfw{{\mathsf{w}}}
\def\sfx{{\mathsf{x}}}
\def\sfy{{\mathsf{y}}}
\def\sfz{{\mathsf{z}}}
\def\sf0{{\mathsf{0}}}

\def\Nt{{N_t}}
\def\Nr{{N_r}}
\def\Ne{{N_e}}
\def\Ns{{N_s}}
\def\Es{{E_s}}
\def\No{{N_o}}
\def\sinc{\mathrm{sinc}}
\def\dmin{d^2_{\mathrm{min}}}
\def\vec{\mathrm{vec}~}
\def\kron{\otimes}
\def\Pe{{P_e}}
\newcommand{\expeq}{\stackrel{.}{=}}
\newcommand{\expg}{\stackrel{.}{\ge}}
\newcommand{\expl}{\stackrel{.}{\le}}
\def\SIR{{\mathsf{SIR}}}

\def\nn{\nonumber}

%


\newlength{\figurewidth}\setlength{\figurewidth}{0.6\columnwidth}




\newcounter{one}
\setcounter{one}{1}
\newcounter{two}
\setcounter{two}{2}

\addtolength{\floatsep}{-\baselineskip}
\addtolength{\dblfloatsep}{-\baselineskip}
\addtolength{\textfloatsep}{-\baselineskip}
\addtolength{\dbltextfloatsep}{-\baselineskip}
\addtolength{\abovedisplayskip}{-1ex}
\addtolength{\belowdisplayskip}{-1ex}
\addtolength{\abovedisplayshortskip}{-1ex}
\addtolength{\belowdisplayshortskip}{-1ex}

\begin{abstract}

  Speed scaling for a tandem server setting is considered, where there
  is a series of servers, and each job has to be processed by each of
  the servers in sequence. Servers have a variable speed, their power
  consumption being a convex increasing function of the speed.
  We consider the worst case setting as well as the stochastic
  setting. In the worst case setting, the jobs are assumed to be of
  unit size with arbitrary (possibly adversarially determined) arrival
  instants. For this problem, we devise an online speed scaling
  algorithm that is constant competitive with respect to the optimal
  offline algorithm that has non-causal information. The proposed
  algorithm, at all times, uses the same speed on all active servers,
  such that the total power consumption equals the number of
  outstanding jobs. In the stochastic setting, we consider a more
  general tandem network, with a parallel bank of servers at each
  stage. In this setting, we show that random routing with a simple
  gated static speed selection is constant competitive. In both cases,
  the competitive ratio depends only on the power functions, and is
  independent of the workload and the number of servers.  
\end{abstract}

\maketitle
\section{Introduction}

Starting with the classical work \cite{yao1995scheduling}, the speed
scaling problem has been widely considered in literature, where there
is a single/parallel bank server with tuneable speed, and the canonical problem is
to find the optimal service speed/rate for servers that minimizes a
linear combination of the flow time (total delay) and total energy
\cite{albers2007unitsizeenergy}, called flow time plus energy, where
flow time is defined as the sum of the response times (departure minus
the arrival time) across all jobs.

Many interconnected practical systems such as assembly lines, flow
shops and job shops in manufacturing, traffic flow in a network of
highways, multihop telecommunications networks, and client-server
computer systems, however, are better modelled as network of
queues/servers.
Another important example is service systems with server specific
precedence constraints, where jobs have to be processed in a
particular order of servers. In such systems, for each job, service is
defined to be complete once it has been processed by a subset of
servers, together with a permissible order on service from different
servers.

The simplest such network is a $K$-server tandem setting, where there
are $K$ servers in series, and each object/job has to be processed by
all the $K$ servers in a serial order. With $K$-tandem servers, we
consider the speed scaling problem of minimizing flow time plus
energy, when the speed/service rate of each server is tuneable and
there is an associated energy cost attached to the chosen speed. The
control variables here include \emph{scheduling}, i.e., which job to run 
on each server, and \emph{speed scaling}, i.e., which speed to
operate each server at. In the worst case setting, the arrival sequence is 
arbitrary, and possibly adverserially determined. In this case, the 
performance metric is the competitive ratio, that is defined as the 
maximum of the ratio of the cost of the online
algorithm and the optimal offline algorithm $\opt$ that is allowed to
know the entire input sequence ahead of time, over all possible
inputs. In the stochastic setting, job arrivals occur according to a 
stochastic process. Here, the cost of an algorithm is the sum of the 
steady state averages of response time and energy consumption per job. 
The competitive ratio of an algorithm is in turn the ratio of its cost to that
of the optimal algorithm (that admits the above steady state averages). 
In both settings, the goal is to design algorithms that have a small competitive 
ratio.

\subsection{Related Work}
\subsubsection{Arbitrary Input Case}
With arbitrary input, where job arrival times and sizes are arbitrary (even chosen by an adversary), for speed scaling with a single server or bank of parallel servers, two classes
of problems have been studied: i) unweighted and ii) weighted, where in i)
the delay that each job experiences is given equal weight in the flow
time computation, while in ii) it is scaled by a weight that can be
arbitrary.

The weighted setting is fundamentally different that the
unweighted one, where it is known that constant-competitive online
algorithms are not possible \cite{bansal2009weighted}, even for a
single server, while constant competitive algorithms are known for the
unweighted case, even for arbitrary energy functions, e.g., the best known 
$2$-competitive algorithm
\cite{SpeedScalingOptFairRobust}. For more prior references on speed scaling, we refer the reader to \cite{SpeedScalingOptFairRobust, bansal2009speedconf}.

In addition to a single server, speed scaling problem has also been considered extensively  for a parallel bank of servers, where there is a single queue and multiple servers. With multiple servers, on an arrival of a new job, the decision is to choose which jobs to run on the multiple servers, by preempting other jobs if required, and what speed 
\cite{lam2008competitive, greiner2009bell, gupta2010scalably,
  lam2012improved, gupta2012scheduling}. The homogenous server case
was studied in \cite{lam2012improved, greiner2009bell}, i.e., power usage function 
is identical for all servers, while the heterogenous case was
addressed in \cite{gupta2010scalably, gupta2012scheduling, devanur2018primal}, where
power usage function is allowed to be different for different servers. 


\subsubsection{Stochastic Input Case}
Under stochastic input, research on two tandem servers with variable speed was initiated  in the classical work of \cite{rosberg1982optimal} \cite{hajek1984optimal}, that established that the optimal service rate of the first queue is monotonic in the joint queue state, and is of the bang-bang type. 
These results have also been extended for any number of tandem servers when each server has exponential service distribution \cite{weber1987optimal}.
These type of problems belong to the class of control of Markov decision processes for which general results have been also derived \cite{ghosh1990ergodic}. Typically, in early works, the objective function did not include an energy cost for increasing the speed of the service rate. To reflect the energy cost, \cite{xia2017service} considered the same problem as in \cite{rosberg1982optimal} in the presence of an average power constraint. Analytical results in this area have been limited to structural results, such as the monotonicity results, and that too for special input/service distributions, and no explicit optimal service rates are known. 
In the stochastic setting, with multiple parallel servers, the flow time plus
energy problem with multiple servers is studied under a fluid model
\cite{asghari2014energy, mukherjee2017optimal} or modelled as a Markov
decision process \cite{gebrehiwot2017near}, and near optimal policies
have been derived. 

\subsection{Our work} 
We consider the speed scaling problem in the tandem network setting, where there are multiple servers ($K$) in series. Each external job arrives at server $1$, and is defined to be {\it complete} once it has been processed by each of the $K$ servers in series. Each server has an identical power (energy) consumption function $P(.)$, i.e., if the server speed is $s$, then power consumed is $P(s)$. 
\subsubsection{Arbitrary Input}

We consider the arbitrary input setting, where jobs can arrive at arbitrary time on server $1$, arrival times possibly chosen by an adversary. However, we assume that each job has the same size/or requirement on any of the servers. Even for a single server setting, initial progress was made for unit sized jobs \cite{albers2007unitsizeenergy, lam2008speedunitsize, bansal2008schedulingunitsize, bansal2009speedunitsize, albers2014unitsizespeed}, which was later generalized for arbitrary job size. In the sequel, it is evident that the considered problem is challenging even with unit sized jobs.
The motivation to consider the arbitrary input setting is two-fold : i) that it is the most general, ii) that even if one assumes that the external arrivals to server $1$ have a nice distribution, with speed scaling by each of the server, the internal arrivals (arrivals at server $k$ corresponding to departures from server $k-1$) need not continue to have the same nice  distribution.
Under the arbitrary input setting, we consider the unweighted flow time + energy as the objective function, and the problem is to find an online algorithm with minimum competitiive ratio. 

The proposed algorithm ensures that there is at most one outstanding job with all servers other than server $1$. Let $n^1(t)$ be the number of outstanding jobs with server $1$, and let the 
total number of servers with outstanding jobs (called active) be $A(t)$ excluding the first server. Then the algorithm runs each active server (including server $1$) at the same speed of $P^{-1}\left(\frac{n^1(t)+A(t)+1}{A(t)+1}\right)$. Thus, the total power consumed across all servers is equal to the number of total outstanding jobs plus $1$, that could be spread across servers. 
The main result of this paper is as follows.

\begin{theorem}\label{thm:main} With unit sized jobs, and identical power consumption function $P$ for all servers, the competitive ratio of the proposed algorithm is at most $\left(6+ \left(12/P(s^\star)\right)\Delta(1)\right)$ where $\Delta(1) = P'(P^{-1}(1))$ and $1+P(s^\star) = s^\star P'(s^\star)$. 
For $P(s) = s^\alpha$, $\Delta(1) = \alpha$ and for $\alpha=2$, $P(s^\star)=2$, making the competitive ratio at most $18$.
\end{theorem}

Even though there has been large number of papers on online speed scaling algorithms with a single server or with multiple parallel servers, as far as we know, there is no work on competitive algorithms for a tandem server case for the objective of flow time plus energy. We would like to point that there is work on only energy efficient routing for networks \cite{andrews2010routing,bansal2012multicast} without any delay consideration. 

With a tandem network, the main technical difficulty in obtaining results for flow time plus energy with the arbitrary input case is that the external arrivals happen at the same time for any algorithm and the optimal offline algorithm $\opt$ into server $1$, but because of dynamic speed scaling, the internal arrivals at intermediate servers (departures from previous server) are not synchronized for any algorithm and the  $\opt$ (that has non-causal information about future job arrivals). Thus, a sample path result that is needed in the arbitrary input case is hard to obtain. 

We overcome this difficulty by proposing a potential function that has positive jumps (corresponding to movement of jobs in consecutive servers) in contrast to typical approach of using potential function that has no positive jumps. Consequently, to derive constant competitive ratio results, we upper bound the sum of the positive jumps and relate that to the cost of the $\opt$. Moreover, the potential function  is only a function of the number of jobs with the $\opt$ in the first server and not in any subsequent servers, since controlling and synchronizing the jobs of the algorithm and the $\opt$ in servers other than the first is challenging.
We show in Remark \ref{rem:nontrivialextension}, that a simple/natural extension of the the popular speed scaling algorithm \cite{bansal2009speedconf} for a single server, does not yield any useful bound on the competitive ratio with tandem servers.
Our result is similar in spirit to the results of \cite{gupta2010scalably,devanur2018primal} for parallel servers, where the competitive ratio also depends on $P(.)$. 
Compared to the prior work on speed scaling with single(parallel) server(s) \cite{bansal2009speedconf,gupta2010scalably,devanur2018primal}, we make a non-trivial extension (even though our results require unit sized jobs) and provide constant competitive ratio results for tandem servers, that has escaped analytical tractability for long. 

\subsubsection{Stochastic Input}
In the stochastic setting, we consider a more
  general tandem network, with a parallel bank of servers at each
  stage. The external arrivals to stage $1$ are assumed to follow a Poisson distribution. We consider a simple 'gated' static speed algorithm and random routing among different servers in each stage that  critically ensures that the job arrivals to subsequent stages are also Poisson \cite{Harchol2013}. 
  We show that the random routing and gated static speed policy has a constant competitive ratio that only depends on the power functions, and is independent of the workload and the number of servers. To contrast our work with prior work on stochastic control of tandem servers \cite{rosberg1982optimal,hajek1984optimal}, the novelty of our work is that we are able to give concrete (constant factor) competitive ratio guarantees, while in prior work only structural results were known that too in  the stochastic input setting.

\section{System Model}
Let the input consist of $n$ jobs, where job $j$ arrives (released at)
at time $a_j$ and has work/size $w_j^k$, to be completed on server $k$.  
There are $K$ homogenous
servers in series/tandem, each with the same power function $P(s)$,
 where $P(s)$ denotes the power consumed by any server while running
at speed $s$. Typically, $P(s)= s^\alpha$ with $2\le \alpha\le 3$. 
Each job has to be processed
by each of the $K$ servers, in sequence, i.e. server $k$ can process a job only after it has been completely processed by server $k-1$ and departed from it. Following most of the prior work in this area, we assume that each server has a large enough buffer and no jobs are dropped ever.

The speed $s$ is the rate at which work is executed by any of the
server, and $w$ amount of work is completed in time $w/s$ by any
server if run at speed $s$ throughout time $w/s$.  A job $j$ is
defined to be complete at time $f_j$ on server $k$ if $w_j^k$ amount of work has been
completed for it on server $k$.  
The flow time $F_j$ for
job $j$ is defined as $F_j = f_j-a_j$ ($f_j$ is the completion time of job $j$ on the last ($K^{th}$) server minus the arrival time)
and the overall flow time is $F= \sum_{j} F_j$. From here on we refer
to $F$ as just the flow time.  Note that $F = \int n(t) dt$, where
$n(t)$ is the number of unfinished jobs (spread across possibly different servers) at time $t$.
We denote the corresponding variables for the the optimal offline algorithm $\opt$ by a subscript or superscript $o$.

Let server $k$ run at speed $s_k(t)$ at time $t$. Then the energy cost for server $k$ is defined as $P(s_k(t))$, where $P(.)$ is strictly convex, non-decreasing, differentiable function at $s\in [0,\infty)$. 
Natural example of $P(x) = a+bx^\alpha, \alpha, a,b \ge 1$ clearly
satisfies all these conditions.  Following \cite{bansal2009speedconf},
these special conditions on $P$ can be relaxed completely, without
affecting the results, and more importantly work even when maximum
speed is bounded $s\in [0,\sfB]$ (see
Remark~\ref{rem:boundedspeed}). Total energy cost is
$\sum_{k=1}^KP(s_k(t))$ summed over the flow time.

Choosing larger speeds reduces the flow time, however, increases the energy cost, and the natural objective function that has been considered extensively in the literature is a linear combination of flow time and energy cost, which we define as 
\begin{equation}\label{eq:cost}
C =  \int n(t) dt + \int  \sum_{k=1}^KP(s_k(t))dt.
\end{equation} 
Note that there is no explicit need for considering the weighted combination of the two costs in \eqref{eq:cost} since a scalar multiple can be absorbed in the power function $P(.)$ itself. 

\section{Arbitrary Input}
Any online algorithm only has causal information, i.e., it becomes aware of job $j$ only at time $a_j$. 
Using only this causal information, any online algorithm has to decide at what speed each server should be run at at each time. Let the cost \eqref{eq:cost} of an online algorithm $A$ be $C_A$, and the cost for the $\opt$ that knows the job arrival sequence $\sigma$ (both $a_j$ and $w_j^k$) in advance be $C_{\opt}$. Then the competitive ratio of the online algorithm $A$ for $\sigma$ is defined as 
\begin{equation}\label{eq:compratio}
\sfc_A(\sigma) = \frac{C_A(\sigma)}{C_{\opt}(\sigma)},
\end{equation} 
and the objective function considered in this paper is to find an online algorithm that minimizes the worst case competitive ratio  
$
\sfc^\star = \min_A \max_{\sigma}  \sfc_A(\sigma)
$.

A typical approach in speed scaling literature to upper bound (of $c$) the competitive ratio is via the construction of a potential function $\Phi(t)$ and show 
that for any input sequence $\sigma,$
\begin{equation}\label{eq:mothereq} 
  n(t) + \sum_{k =1}^K P(s_k(t)) +  \frac{d\Phi(t)}{dt}  \le c(n_o(t)
  + \sum_{k=1}^K P(s^o_k(t))),
\end{equation} 
almost everywhere and that $\Phi(t)$ satisfies the following {\it
  boundary} conditions,
\begin{enumerate}
\item Before any job arrives and after all jobs are finished,
  $\Phi(t)= 0$, and
\item $\Phi(t)$ does not have a positive jump discontinuity at any
  point of non-differentiability.
\end{enumerate} 
Then, integrating \eqref{eq:mothereq} with respect to $t$, we get that
\begin{equation*}\label{eq:mothereq1} \left(\int n(t) + \sum_{k =1}^K
  P(s_k(t))\right) \le \int c\biggl(n_o(t) + \sum_{k =1}^K P(s^o_k(t))\biggr),
\end{equation*} 
which is equivalent to showing that $C_{A}(\sigma) \le c \
C_{\opt}(\sigma)$ for any input $\sigma$ as required. 

Since any online algorithm is only allowed to make causal decisions, thus at any time $t$, the speed chosen by an online algorithm $A$ for any server and the $\opt$ can be different. Because of this, 
the main challenge when there are tandem servers, is that the internal arrivals at server $k+1$ that corresponds to departures from  server $k$ ($k< K$ other than the last)  can happen at different times for the algorithm and the $\opt$. Thus constructing a potential function and ensuring that the boundary conditions are satisfied presents a unique challenge. With a single (or parallel bank) server such a problem does not arise since there, arrivals only happen externally at the same time for both the algorithm and the $\opt$. Thus, instead of finding a potential function that does not have a positive jump discontinuity, we propose a potential function for which  we can control how large the positive jump discontinuity and compare it with cost of the $\opt$. Let the new boundary conditions be,  \begin{enumerate}
\item Before any job arrives and after all jobs are finished,
  $\Phi(t)= 0$, and
\item Let $\Phi(t)$ increase by amount $D_j$ at the $j^{th}$ discontinuous point. Let $\sum_{j} D_j \le \sfD C_{\opt}$.
\end{enumerate} 
Then, integrating \eqref{eq:mothereq} with respect to $t$, we get that
\begin{align}\label{eq:mothereq22} 
C_A \le & \int \left(n(t) + \sum_{k =1}^K
  P(s_k(t))\right) dt + \int \frac{d\Phi(t)}{dt} dt\le \int c\biggl(n_o(t) + \sum_{k =1}^K P(s^o_k(t))\biggr)dt + \sfD C_{\opt},
\end{align} 
which is equivalent to showing that $C_{A}(\sigma) \le (c +D)\
C_{\opt}(\sigma)$ for any input $\sigma$ as required. 

The main novel contribution of this paper is the construction of a potential function for tandem server settings with positive jumps, where we can upper bound $\sfD$, and importantly which is only a function of the number of jobs with the $\opt$ on the first server (which arrive together for the algorithm as well) and not on subsequent servers, since controlling them is far too challenging.


{\bf Job sizes:} For a single server setting,  constant competitive algorithms have been derived independent of the job sizes \cite{bansal2009speedconf}. Considering arbitrary job sizes in a multiple tandem server setting is more complicated (technical difficulty is described in Remark \ref{rem:nontrivial}) and we consider homogenous job size setting, where all job sizes are identical across all jobs and all servers $w_j^k= w, \ \forall \ j, k$. Without loss of generality, we infact let $w=1$.


We here discuss briefly why it is non-trivial to extend the results for single server setting to tandem server setting.
\begin{remark}\label{rem:nontrivialextension} Let $w=1$.
Consider a $c$-competitive algorithm $A_c$ for a single server with equal job size, e.g. $c=3$ \cite{bansal2009speedconf} that chooses speed $s = P^{-1}(n+1)$, where $n$ is the number of outstanding jobs. There are two ways to use this in the 
tandem server setting. Let $n^i$ be the number of jobs on server $i$. Either we can replicate the speed of jobs as seen on server $1$ ($s_1 = P^{-1}(n^1+1)$) on server $2$, or use $s_i = P^{-1}(n^i+1), i=1,2$ autonomously on both the servers. We argue next that both these choices are not very useful.

{\bf Speed Replication:} Let job $j$ arrive at time $a_j$ and depart at $f_j$, and during this time the speed chosen by server $1$ to serve job $j$ be $s_j, t\in [f_j, a_j]$. 
Replicating the speed profile $s_j$ on the second server as well does not result in $2c$-competitive algorithm for the two-server problem. 
What can happen is that consider a time $t$ where a job $j$ starts its service at server $2$ and let that job $j$ was alone in server $1$ throughout the time it spent in server $1$, i.e. its speed profile $s_j = P^{-1}(2), t\in [f_j, a_j]$ \cite{bansal2009speedconf}. Let the next job $j+1$ arrive at $t=f_j$ into server $1$. 
The speed of job $j+1$ in server $1$ is $P^{-1}(2)$, and because of replication of $s_j$ on server $2$ for job $j$, job $j$ is also being processed at speed $P^{-1}(2)$. Let at $t^+$, $n >>1$ new jobs arrive in server $1$, because of which the speed of job $j+1$ is increased to $P^{-1}(n+2)$. Thus, with the $2$-server setting, job $j+1$ will be processed fast and will have to wait behind job $j$ in server $2$ since job $j$'s speed is fixed at $P^{-1}(2)$. Such a problem is avoided in a single server system since at time $t^+$ job $j$ has departed the system. Thus, with the two-server system, the cost for job $j+1$ could be more than two times compared to a single server system. 

{\bf Autonomous: } Consider an input, where $\ell > >1$ jobs of unit size arrive at time $0$ into server 
$1$. 
Then choosing $s_i = P^{-1}(n^i+1), i=1,2$, server $2$ runs slower compared to server $1$ until 
$\ell/2$ jobs have been processed by server $1$ and are available at server $2$. Thus, jobs start accumulating in server $2$'s queue, and consequently, each of $\ell$ jobs have to wait behind jobs in server $2$ for sufficient time before they are processed by server $2$, entailing a large flow time + energy cost. This argument on its own does not mean that the competitive ratio of this algorithm is poor, since the inherent cost could be large even with the $\opt$ with this input. However, for this input, instead a simple algorithm ($\opt$ can only do better) that chooses $s_i = P^{-1}(n^1+1)$ for both $i=1,2$ avoids any waiting for 
any job on server $2$ and can be shown to have at most twice the flow time + energy cost of the server $1$. 
Thus autonomous speed choice for two servers is also not expected to provide low (or constant) competitive ratio.


\end{remark}


\subsection{Speed Scaling Algorithm}\label{sec:hom}
We begin this section, by first deriving a lower bound on the cost of the $\opt$.
\begin{lemma}\label{lem:enhOPT} $C_{\opt} \ge C_{\opt-E}$, where 
$ C_{\opt-E} = \int \left(n^1_o(t) + \sum_{k=1}^K P(s^1(t))\right) dt$.
\end{lemma}
\begin{proof}
We enhance the $\opt$ as follows to derive a lower bound on its cost. Instead of requiring that $\opt$ processes jobs in series, each incoming job is copied on all servers and a job is defined to be complete, when it is completed by all servers. Thus, allowing $\opt$ to run jobs in parallel. Essentially this will let $\opt$ run jobs at same speed in each of the servers, and have the same number of outstanding jobs on each server. Thus, for the enhanced $\opt$, the total cost (flow time + energy) is equal to $ C_{\opt-E} = \int \left(n^1_o(t) + \sum_{k=1}^K P(s^1(t))\right) dt$, where $n^1_o(t)$ is the number of outstanding jobs on server $1$. Thus, we have $C_{\opt} \ge C_{\opt-E}$.
\end{proof}

Next, we will compare the performance of the proposed algorithm and the enhanced $\opt$.
Let $n^k(t)$ and $n_o^k(t)$ the number of outstanding jobs on server
$k$ with the algorithm and the enhanced $\opt$ (which for succinctness call $\opt$ whenever there is no ambiguity), respectively at time $t$. For the enhanced $\opt$ we only need to consider the number of jobs on server $1$. 
At time $t$, let $n_o(t,q)$ be the number of unfinished
jobs with $\opt$ on the first server with remaining size at least~$q$, while $n^i(t,q)$ be the number of unfinished
jobs with the algorithm on server $i$ with remaining size at least~$q$. Thus, $n^i(t) = n^i(t, 0)$ and $n_o(t)
= n_o(t, 0)$.


For server $1$, let $d^1(t,q) = \max\left\{0, \frac{n^1(t,q) -
  n_o(t,q)}{K}\right\}$, while for server $j\ge 2$,
$$d^j(t,q)=   \sum_{k=2}^{j-1} n^k(t) + n^j(t,q),$$
where $\sum_{k=2}^{j-1} n^k(t)$ is  the total number of outstanding jobs from server $2$ till server $j-1$. Notably in defining $d^j(t,q)$ there is no contribution from the $\opt$ unlike in $d^1(t,q)$. This is key, since there is no way to control the number of jobs that the $\opt$ has in server $j\ge 2$ and their transitions between servers $j$ to $j+1$.

For
the algorithm, a server~$i$ is defined to be {\it active} if it has an
unfinished job, i.e., $n^i(t) > 0.$ The indicator function $A_i(t)=1$ for $i\ge 2$
if server~$i$ is active under the algorithm at time $t$ and $A_i(t)=0$
otherwise. Then $A(t) = \sum_{i=2}^K A_i(t)$ is the number of active
servers with the algorithm at time~$t$, other than server $1$.

For $i \in \N,$ let
$$f_a\left(\frac{i}{a}\right)-f_a\left(\frac{i-1}{a}\right) =
\Delta\left(\frac{i}{a}\right) :=
P'\left(P^{-1}\left(\frac{i}{a}\right)\right)$$ and $f_a(0) = 0$. 
For server $1$, let 
\begin{equation}
  \label{eq:phi1}
  \Phi_1(t) = 
c \int_{0}^1  f_1\left(d^1(t,q)\right)dq ,
\end{equation}
while for server $j\ge 2$,
\begin{equation}
  \label{eq:phij}
  \Phi_j(t) = 
 \Phi_1(t)  + \Phi_j^{\alg}(t)
\end{equation}
where $$\Phi_j^{\alg}(t)= c_j \int_{0}^1
f_{A(t)+1}\left(\frac{d^j(t,q)}{A(t)+1}\right) dq. $$ 

Consider the potential function
\begin{equation}\label{defn:phi}
  \Phi(t)=\sum_{j=1}^{K}\Phi_j(t),
\end{equation} 

{\bf Algorithm: } The speed scaling algorithm that we propose, chooses
the following speeds. For server $1$, 
\begin{equation}\label{eq:speeddef}
  s_1(t) = \begin{cases}   P^{-1}\left(\frac{n^1(t)+A(t)+1}  {A(t)+1}\right), & \text{if} \ n^1(t)> 0, \\
 0, &\text{otherwise.}\end{cases}
 \end{equation}
For active servers, i.e., servers  $i\ge 2$  with  $A_i=1$
 \begin{equation}
 s_i(t) = \begin{cases}   P^{-1}\left(\frac{n^1(t)+A(t)+1}  {A(t)+1}\right), & \text{if} \ n^1(t)> 0, \\
 P^{-1}(2), &\text{otherwise.}\end{cases}
 \end{equation}
 The non-active servers have speed $0$.

 With this speed scaling choice,
 all active servers work at the same speed at each time, and since we are assuming
 that each job has the same size on all servers, this implies that
 jobs only wait in server $1$ if at all, and are always in process at
 active servers $i>1$. Moreover, other than server $1$, all servers have at most $1$ outstanding job at any time. The speed choice ensures that the total power used is $n^1+A+1$ (or $2A$ if $n^1=0$) one more than the total number of outstanding jobs in the system.
 
{\bf Comments about the potential function:} The basic building blocks $\Phi_1$ and $\Phi_j^\alg$ of our proposed potential function are inspired by the potential function first constructed in \cite{bansal2009speedconf}, however, the non-trivial aspect is the choice of including $A(t)$ to define the $f$ function. Since $A(t)$ changes dynamically, the overall construction and analysis is far more challenging.

The proposed potential function $\Phi$ is rather delicate and is really the core idea for solving the problem. We discuss its important properties and reasons why a more natural choice does not work as discussed in Remark \ref{rem:nat}. To begin with, note that the denominator in $d^1(t,q)$ is fixed to be $K$ and not $A(t)$ which can change dynamically. This is important since $n^1(t)$ can be arbitrarily large, and a decrease in $A(t)$ can have an arbitrarily large increase in $\Phi_1(t)$. However for $d^j(t,q)$ which is function of $A(t)$, even when $A(t)$ decreases, the increase in $\Phi_j(t)$ can be bounded since $n^j(t)\le 1$ (choice made by the proposed algorithm) and $\sum_{j\ge 2} n^j(t) \le A(t)$. The choice of potential function is also peculiar since $\Phi_1(t)$ is spread over all the $K$ servers with a normalization factor of $K$ (as defined in $d^1(t,q)$). This is needed since the algorithm prescribes an identical speed 
of $P^{-1}\left(\frac{n^1(t)+A(t)+1}{A(t)+1}\right)$ for all the servers, and to get sufficient negative drift from the $d\Phi_1(t)/dt$ term, it is necessary that $P\left(P^{-1}\left(\frac{n^1(t)+A(t)+1}{A(t)+1}\right)\right) \ge n^1/K$, which is true since $A(t) \le K-1$. If instead we just keep one term for $\Phi_1(t)$ in $\Phi(t)$ without the normalization by $K$ in $d^1(t,q)$, the speed of each server has to be at least $P^{-1}(n^1(t))$ to get sufficient negative drift from the $d\Phi_1(t)/dt$ term, however, that makes the total power used $\sum_{j=1}^nP^{-1}(s_j) = K n^1(t)$, which is order wise too large.

\begin{remark}\label{rem:nat} The considered potential function \eqref{eq:phij} for server $j$ is not a natural choice. Instead it should really be $$\Phi_j(t) = c_j \int_{0}^1
f_{A(t)+1}\left(\frac{ n^1(t) + \sum_{k=2}^{j-1} (n^k(t))+ n^j(t,q) }{A(t)+1}\right) dq,$$
by combining the arguments of $\Phi_1$ and $\Phi^{\alg}_j$ into a single $f$ function. 
This choice avoids the increase in $\Phi_j$ when a job moves from server $k$ to $k+1$ unlike \eqref{eq:phij}, since in this case $n^k(t^+) = n^k(t)-1$, while $n^{k+1}(t^+,q) = n^{k+1}(t,q)+1$ cancelling each other off. 
This, however, makes controlling the increase in $\Phi_j(t)$ when $A(t)$ decreases, since $n^1(t)$ can be arbitrarily large. 
The current choice \eqref{eq:phij} avoid this bottleneck by isolating server $1$ from all the other subsequent servers by keeping the terms of server $1$ and subsequent servers \eqref{eq:phij} separate, however, at a cost of incurring positive jumps whenever jobs move from server $k$ to $k+1$ which can be bounded.
\end{remark}
\begin{remark}
To eliminate the need for considering different epochs at which the job transition happens between server $k$ and server $k+1$ with the algorithm and the $\opt$ which can result in increase in the potential function, one can consider a following equivalent model. Let on (external) arrival of a new job $j$ to server $1$ at time $t$, $K$ jobs are created with sizes $w$, and the $k^{th}$ copy with size $w$ is sent to the $k^{th}$ server at time $t$. To model the tandem server constraint, a {\bf precedence} constraint can be enforced such that any copy of any job cannot start its processing at server $k$ unless it has been processed (served and departed) at the server $k-1$. The precedence constraint, however, brings in a new feature unlike the single server case, that the servers can idle even when they have outstanding jobs, if those jobs have not been processed by preceding servers, which needs to be handled carefully.

Following \cite{bansal2009speedconf}, a natural choice for the potential function with this alternate model is 
$\Phi_1(t) =  c_1\int_{0}^\infty f\left(d^1(t,q)\right) dq$, and 
$\Phi_k(t) =  c_k\int_{0}^\infty f\left(d^k(t,q)\right) dq$, and consider the potential function 
$
\Phi(t) = \Phi_1(t)+\sum_{k=2}^K\Phi_k(t)$, 
where  $d^1(t,q) = \max\left\{0, n^1(t,q) - n^1_o(t,q)\right\}$, and 
 \begin{align*}
 d^k(t,q)& = \max\left\{0, (n^k(t)(t,q) - n^{k-1}(t,q)) -  (n^k_o(t,q) - n^{k-1}_o(t,q))) \right\}.\end{align*}
 and $f(0) = 0$, and $\forall \ i\ge 1$, $ f(i)-f(i-1) = \Delta(i) :=
P'(P^{-1}(i))$ (this means $P'(x)$ where $x=P^{-1}(i)$).
To get the correct negative drift with this potential function, however, requires $c_k > c_{k+1}$ because of the 'back flow' (terms of type $n^k(t,q) - n^{k-1}(t,q)$ in $\Phi_k$ which increase the potential function $\Phi_k$ when the algorithm is working on server $k-1$) making $c_1\ge K$, and since the competitive ratio at least $c_i$ for all $i$, the resulting competitive ratio turns out to be $K$.
\end{remark}

From here on we work towards proving Theorem \ref{thm:main}. The first step in that direction is to bound the increase in the potential function $\Phi(t)$ at discontinuous points, which is done as follows.
  
\begin{lemma}\label{lem:jump}
  Taking $c_j = c$ for all $j\ge 2$, the total increase in $\Phi(t)$ at
  points of discontinuity is at most $2 c nK \Delta(1)$.
\end{lemma}
Proof  can be found in 
Appendix~\ref{app:lem:jump}. 
\begin{remark}\label{rem:nontrivial} The restriction of equal job sizes is essentially needed to prove Lemma \ref{lem:jump}. 
Since all server speeds are identical, if job sizes are different, jobs will accumulate in servers other than $1$, making it hard to control the increase in $\Phi(t)$ when $A(t)$ decreases.
\end{remark}
\begin{definition}\label{defn:r}
Let $r_i=\ell$, if $i$ is the $\ell^{th}, \ell \in [1:A(t)]$ active server (in increasing order of server index) among the $A(t)$ active servers.
\end{definition} 
The proof of Theorem~\ref{thm:main} is based on the following bounds
on the potential function drift. 
\begin{lemma}
\label{lem:bansalphi1}
Consider any instant $t$ when no arrival/departure (including internal
transfers) occurs under the algorithm or $\opt.$ 
For server $1$, 
\begin{align*}
  d\Phi_1/dt\le \begin{cases} cP(s_o) - c\frac{(n_1-n_o)}{K} & \text{if} \ n_o < n^1,
 \\ 0 &\text{if} \  n_o > n^1.
 \end{cases}
 \end{align*}
 If $n_o(t) = n^1(t)$, then either of the above two cases arise. Moreover,  for any active
server~$i\ge 2$ at time $t$,
\begin{align*}
  d\Phi_i^\alg/dt\le 
 - A_i(t) c_i \frac{r_i}{A(t)+1},
\end{align*}
where $A_i(t)=1$ when server $i\ge 2$ is active and zero otherwise. 
\end{lemma}
Proof  can be found in 
Appendix~\ref{app:lem:drift}. 

  Next, we consider the cost of the algorithm at any time $t$, and suppress $(t)$ for simplicity.
  When $n_o > n^1$ and $n^1\ne 0$, then $d\Phi_1/dt= 0$ (Lemma \ref{lem:bansalphi1}),  
   and since $\sum_{i=2}^K\frac{r_i}{A+1} \ge A/2$, the `running' cost \eqref{eq:mothereq} from Lemma \ref{lem:bansalphi1} with $c_j=c, \ \forall \ j\ge 2$ is 
   \begin{align*} n^1+A + \sum_{k=1}^K P(s_i(k)) +d\Phi/dt & \le n^1+A + n^1+A +1- c(A/2)  \le 3n_o,
   \end{align*} choosing $c=6$.
   If $n^1=0$, where each active server other than server $1$ has speed $P^{-1}(2)$, then the running cost  
    \begin{align*} A + \sum_{k=1}^K P(s_i(k)) +d\Phi/dt & \le A + 2A - c(A/2) \le0,
   \end{align*} choosing $c=6$.
  When $n_o < n^1$, then $d\Phi_1/dt \le P(s_o) - c\frac{(n_1-n_o)}{K}$ using Lemma \ref{lem:bansalphi1}. Moreover, from Lemma \ref{lem:bansalphi1} with $c_j=c, \ \forall \ j\ge 2$, where $\sum_{i=2}^K\frac{r_i}{A+1} \ge A/2$, the running cost,  
  \begin{align*} 
 n^1+A + \sum_{k=1}^K P(s_i(k)) +d\Phi/dt 
 &\le n^1+A + n^1+A + 1+ \sum_{k=1}^K \left(cP(s_o) - \frac{c(n_1-n_o)}{K}\right) - c(A/2), \\
 & \le c n_o + (2-c) n^1 +A(2-c/2) + 1+ c\sum_{k=1}^K P(s_o),\\
 & \le 6n_o + 6\sum_{k=1}^K P(s_o),
 \end{align*} choosing $c=6$. 
 Thus, in both cases, accounting for the discontinuities from Lemma \ref{lem:jump} with $c_j=c=6$ for all $j$ since the first boundary condition is trivially met, 
 $$\int n^1+A + \sum_{k=1}^K P(s_i(k)) +d\Phi/dt \le 6\left(\int \left(n_o+ \sum_{k=1}^K P(s_o) \right)dt\right) + 12\Delta(1)(nK),$$
 which implies that 
 \begin{equation}\label{eq:finalbound}C_A \le 6C_{\opt-E}+  12\Delta(1)nK.
 \end{equation}

%


Now we complete the Proof of Theorem \ref{thm:main}.
\begin{proof} From \eqref{eq:finalbound}
\begin{align}\label{eq:res1}
C_A \le 6 C_{\opt-E} +12 \Delta(1) n K.
\end{align}
Recall that $C_{\opt-E}\le C_{\opt}$.
For any job with size $w$, the
minimum cost incurred by $\opt$ on processing it on any one server is
$\min_s \frac{w}{s}+\frac{wP(s)}{s}$, where $s$ is the speed. Thus,
the optimal $s^\star$ satisfies $1+P(s^\star)= s^\star P'(s^\star)$,
and the optimal cost is $wP'(s^\star)$.  With $n$ jobs arriving
each with size $1$ which have to be processed by each of the $K$
servers, a simple lower bound on the cost of $\opt^1$ is $nKP'(s^\star)$.
This
implies from \eqref{eq:res1} that
\begin{align}\label{eq:res11}
C_A & \le 6C_{\opt} + \left(12/P'(s^\star)\right) \Delta(1) C_{\opt} =C_{\opt}(6+\left(12/P'(s^\star)\right)\Delta(1)).
\end{align}
For $P(s) = s^2$, $s^\star
=1$ and the minimum cost is $P'(s^\star)=2$, and $\Delta(1)=2$, thus $C_A \le 18C_{\opt}$.

\end{proof}

\section{Stochastic setting}

In this section, we move from the worst case setting to the stochastic
setting, where the workload is specified by stochastic processes and
we evaluate algorithms based on their performance in steady state. We
find that the stochastic setting is `easier' than the worst case
setting; specifically, we show that a naive routing strategy coupled
with a simple ON/OFF (gated static) speed selection is constant
competitive. Crucially, the competitive ratio depends only on the
power functions, and not on the statistical parameters of the workload
or the topology of the queueing system. Moreover, the tandem system we
consider in this section is more general that the one considered
before---each job needs to be served in $K$ tandem layers/phases,
where each layer~$i$ is composed of $m_i$ parallel servers.

Formally, our system model is as follows. The service system is
composed of $K$ tandem layers of servers, with $m_i$ parallel and
identical servers in layer~$i.$ Jobs arrive to layer $1$ according to a Poisson
process with rate $\lambda.$ The jobs have to be processed
sequentially in the $K$ layers (by any server in each layer) before
exiting the system. Moreover, we assume that each server is equipped
with an (infinite capacity) queue, so that once a job completes
service in layer~$i,$ $1 \leq i \leq K-1,$ it can be immediately
dispatched to any server in layer~$i+1.$ The service requirement in
layer~$i$ is exponentially distributed with mean $1/\mu_i.$ Job
scheduling on each server is assumed to be blind to the service
requirements of waiting jobs. The power function for all servers in
layer~$i$ is $P_i(s) = c_i s^{\alpha_i},$ where $c_i > 0,$
$\alpha_i > 1.$

The performance metric is given by $$C = \Exp{T} + \Exp{E},$$ where
$T$ and $E$ denote, respectively, the response time and energy
consumption associated with a job in steady state. We note that the
performance metric can be also be expressed as the sum of the costs
incurred in each layer:
$$C = \sum_{i = 1}^K \left( \Exp{T_i} + \Exp{E_i}\right).$$
Here, $T_i$ denotes the steady response time in layer~$i,$ and $E_i$
denotes the steady state energy consumption (per job) in
layer~$i.$\footnote{We implicitly restrict attention to the class of
  policies that admit these stationary averages.}

The proposed algorithm ($A$) is the following. When a job arrives into
layer $i,$ we dispatch it to a random server in layer~$i,$ chosen
uniformly at random. The speed of each server in layer $i$ is set in a
\emph{gated static} fashion as $s_i^A = 1 + \frac{\rho_i}{m_i}$ when
active (and zero when idle), where $\rho_i := \frac{\lambda}{\mu_i}$
is the offered load to layer $i.$ Note that the speed selection
requires knowledge (via learning if not available) of the offered load
into each layer (unlike the dynamic speed scaling algorithm in
\eqref{eq:speeddef}. This boils down to learning the arrival rate and
the mean service requirement, which is feasbile in the stochastic
workload setting considered here. Under the proposed random routing
and gated static speed selection, the system operates as a
(feedforward) Jackson network, with each server operating as an M/M/1
queue \cite{Harchol2013}. Thus, the arrival process for each layer is also Poisson.

Our main result is the following. Let $[K] := \{1,2,\cdots,K\}.$
\begin{theorem}
  \label{thm:stochastic:constant_comp}
  The algorithm $A$ is constant competitive, with a competitive ratio
  that depends on only the power functions, i.e., on
  $\bigl((c_i,\alpha_i):\ i \in [K]\bigr)$. Specifically, the
  competitive ratio does not depend the workload parameters $\lambda,$
  $(\mu_i, i \in [K])$, the number of layers $K,$ or on the number of
  servers in the different layers $(m_i, i \in [k]).$
\end{theorem}

The proof of Theorem~\ref{thm:stochastic:constant_comp} can be found
in Appendix~\ref{app:stochastic}.

\section{Proof of Lemma \ref{lem:jump}}
\label{app:lem:jump}
There are 4 possible ways that can give rise to a discontinuity in
$\Phi(\cdot).$ 
\begin{enumerate} 
\item \emph{Job arriving at server $1$ at time $t$.} On arrival of a new job which happens only on server $1$,  
both $n^1(t,q)$, and
  $n_o(t,q)$ increase by $1$ for all $q\in [0,1]$. Hence, there is no change to the $\Phi(t)$ in this case.

\item \emph{Transfer of jobs between servers under the algorithm
    (without departure from server $K$) at time $t.$} 
  For each job transitioning from server $i$ to $i+1$ for $i\ge 1$, there is potentially a positive jump in $\Phi(t)$ because  of either increase in $n^i(t)$ or $n^i(t,q)$ for $i\ge 2$. In particular, for $n$ jobs, there are at most $n$ jumps in $\Phi^\alg_j(t)$ for $j=2, \dots, K$, with each jump of size at most 
\begin{align*}
    \Phi_j^\alg(t^+) - \Phi_j^\alg(t) &= c_j \int_{0}^1
    \left\{f_{A(t)+1}\left(\frac{d^j(t^+,q)}{A(t)+1}\right) -
      f_{A(t)+1}\left(\frac{d^j(t,q)}{A(t)+1}\right) \right\} dq, \\
    &\leq c_j \int_{0}^1
    \left\{f_{A(t)+1}\left(\frac{d^j(t,q)+1}{A(t)+1}\right) -
      f_{A(t)+1}\left(\frac{d^j(t,q)}{A(t)+1}\right) \right\} dq, \\
    &= c_j \Delta\left(\frac{d^j(t,q)+1}{A(t)+1} \right), \\
    &\leq c_j \Delta(1),
  \end{align*}
  since $d^j(t,q)\le A(t)$.
  Counting for at most $nK$ such jumps, the total increase in $\Phi(t)$ is $\Delta(1)n\sum_{j=2}^K c_j$. Note that for each jump either $A(t)$ remains same or increases by $1$. In the above bounding we have taken the worst case, when $A(t)$ remains the same. If $A(t)$ increases by $1$, then the same bound follows using second part of Lemma \ref{lemma:f_prop2}.
  Note that transfer of jobs between servers under the $\opt$ without any departure from server $K$ has no effect on $\Phi(t)$.

%
%
%

\item \emph{Job departing from server $K$ under algorithm at time
    $t.$} 
  We consider two subcases. If $n^1(t) \leq 1,$ then $A(t^+) =
  A(t)-1.$ In this case, 
  \begin{align*}
    \Phi_j^\alg(t^+) - \Phi_j^\alg(t) &\leq c_j \int_{0}^1
    \left\{f_{A(t)}\left(\frac{d^j(t^+,q)}{A(t)}\right) -
      f_{A(t)+1}\left(\frac{d^j(t,q)}{A(t)+1}\right) \right\} dq \\
    & \stackrel{(a)}\leq c_j \int_{0}^1
    \left\{f_{A(t)+1}\left(\frac{d^j(t^+,q)+1}{A(t)+1}\right) -
      f_{A(t)+1}\left(\frac{d^j(t,q)}{A(t)+1}\right) \right\} dq \\
    &= c_j \Delta\left(\frac{d^j(t^+,q)+1}{A(t)+1} \right) \\
    &\stackrel{(b)}\leq c_j \Delta(1).
  \end{align*}
  Here, ($a$) follows from first part of Lemma~\ref{lemma:f_prop2}, while ($b$) is a
  consequence of: $$d^j(t^+,q) \leq A(t).$$

  On the other hand, if $n^1(t) > 1,$ then $A(t^+) = A(t).$ In this
  case, it is easy to see that $\Phi_j(t^+) - \Phi_j(t) \leq 0.$

  Thus, the departure of a job from the system under the algorithm can
  result in an upward jump in $\Phi(\cdot)$ of at most $\Delta(1)
  \sum_{j = 1}^K c_j.$ Choosing $c_j=c$ for all $j$, the total increase in $\Phi(t) \le n\Delta(1)\sum_{j = 1}^K c_j$.

\item \emph{Completion of jobs by $\opt.$} Any job completed by $\opt$ on server changes $n_o(q)$ only for $q=0$ thus, keeping the integral to define $\Phi_i(t)$ unchanged for all $i$.

\end{enumerate}

\begin{lemma}
  \label{lemma:f_prop2}
  For $a,d \in \N$  where $d \leq a,$ $$f_{a}\left(\frac{d}{a}\right) \leq f_{a+1}\left(\frac{d+1}{a+1}\right).$$

  Moreover, for $a,d \in \N,$  $$f_a\left(\frac{d}{a}\right) \geq f_{a+1}\left(\frac{d}{a+1}\right).$$ 
\end{lemma}
\begin{proof}
  To prove the first statement, we note that
\begin{align*}
f_{a}\left(\frac{d}{a}\right) &= \sum_{j = 1}^d \Delta(j/a) \\
& \stackrel{(a)}{\leq} \sum_{j = 1}^d \Delta(j+1/a+1) = \sum_{j = 2}^{d+1} \Delta(j/a+1) \\
& \leq \sum_{j = 1}^{d+1} \Delta(j/a+1) = f_{a+1}\left(\frac{d+1}{a+1}\right).
\end{align*}
Here, $(a)$ follows from the monotonicity of $\Delta(\cdot).$ The second statement of the lemma is trivial: 
$$f_a\left(\frac{d}{a}\right) = \sum_{j = 1}^d \Delta(j/a) \geq \sum_{j = 1}^d \Delta(j/a+1) = f_{a+1}\left(\frac{d}{a+1}\right).$$
\end{proof}


\section{Proof of Lemma \ref{lem:bansalphi1}}\label{app:lem:drift}
Our proofs will require the following technical Lemma from
\cite{bansal2009speedconf}.
\begin{lemma}\label{lem:bansal}[Lemma 3.1 in~\cite{bansal2009speedconf}]
  For $s, {\tilde s}, \beta \ge 0$, then for any function $P$ that is
  strictly increasing, strictly convex, and differentiable,
\begin{align*}
\Delta(\beta)(-s + {\tilde s}) \le& \left(-s +P^{-1}(\beta)\right)P^{'}(P^{-1}(\beta)) 
 + P({\tilde s}) -i.
\end{align*}
\end{lemma}
Proof of Lemma \ref{lem:bansalphi1}.
\begin{proof}
  Since the statement of the lemma applies to a fixed (though generic)
  time~$t,$ we shall omit the reference to $t$ throughout this proof
  for notational simplicity. Let $q_i$ and $q_o$ be the size of the job under process at server $i$ with the algorithm, and with the $\opt$ on server $1$, respectively. Recall that the speed of all active servers with the algorithm is $s_i = P^{-1}\left(\frac{n^1+A+1}{A+1}\right)$, while the speed of server $1$ with the $\opt$ is $s_o$.
  
The main idea of bounding $d\Phi/dt$ is similar to \cite{bansal2009speedconf} being specialized for this potential function and the speed choice.

Case 1: If $n_o > n^1$, then we first show that $d\Phi_1/dt\le0$. Note that under this condition, 
$n_o(q) > n^1(q)$ for $q\in [q_o-s_o dt, q_o]$. Thus, at time $t+dt$, 
 $n_o(q)$ is still at least as much as  $ n^1(q) $ for $q\in [q_o-s_o dt, q_o]$. Therefore, $\Phi_1$ does not increase because of processing by $\opt$. Since processing by the algorithm can only reduce $\Phi_1$, thus,  $d\Phi_1/dt\le 0$.

Case 2: If $n_o < n^1$ and $n^1 > 0$ (since otherwise again $d\Phi_1/dt=0$). Because of processing of jobs by the algorithm and the
  $\opt$, $d\Phi_1/dt$ changes because of reduction in $n^1(q)$
  (because of the algorithm) and $n_o(q)$ (because of the $\opt$). Then for the algorithm, $n^1(q)$
  decreases by $1$ for $q \in [q^1 - s^1 dt, q^1]$, and the
  contribution in $d\Phi_1/dt$ because of the algorithm is
  \begin{align}\label{eq:dummy1}
    - c \Delta \left( \frac{ n^1(q^1) -1 + n_o(q^1)}{K}\right)s_i,
\end{align}
where $\Delta$ has been defined after \eqref{defn:phi}.

Similarly, for the $\opt$ $n_o(q)$ decreases by $1$ for $q \in [q_o
- s_o dt, q_o]$, and the contribution in $d\Phi_1/dt$ because
of the $\opt$ is
\begin{align}\label{eq:dummy2}
  c \Delta \left( \frac{ n^1(q_o) - n_o(q_o)+1}{K}\right) s_o.
\end{align}
As shown in \cite{bansal2009speedconf}, that the argument of $\Delta(\cdot)$ is equal
in \eqref{eq:dummy1} and \eqref{eq:dummy2}. 
Combining
\eqref{eq:dummy1} and \eqref{eq:dummy2}, we get that
\begin{align}
d\Phi_1/dt & =  c \Delta \left(\frac{n^1 - n_o}{K}\right)(-s_1 + s_o) \label{eq:dummy3}
\end{align} or 
\begin{align}
d\Phi_1/dt & =  c \Delta \left(\frac{n^1 - n_o+1}{K}\right)(-s_1 + s_o),  \label{eq:dummy309}
\end{align}
depending on whether $q_1 > q_o$ or otherwise.
For either case, we apply technical Lemma \ref{lem:bansal}, to bound RHS of
\eqref{eq:dummy3} (similar bound will work for \eqref{eq:dummy309} as well). Setting $\beta = \frac{n^1 - n_o}{K},$ using Lemma \ref{lem:bansal}, note that
\begin{align}\nn
\Delta \left(\beta\right)(-s^i + s_o)& \le  \left(-s^i +P^{-1}(\beta)\right)P^{'}(P^{-1}(\beta)) + P(s_o) -\beta, \\ 
& \stackrel{(a)}\le  P(s_o) -\beta, \label{eq:dummy456}
\end{align}
where $(a)$ follows from the speed definition \eqref{eq:speeddef} $s_1 = P^{-1}\left(\frac{n^1+A+1}{A+1}\right)$
which ensures that $P(s^i) \geq \beta$, since $A \le K-1$. Thus, 
\begin{align}
d\Phi_1/dt & =  c ( P(s_o) -\beta). \label{eq:dummy33}
\end{align}
If $n^1=n_o$ then either  $d\Phi_1/dt=0$ or  \eqref{eq:dummy33} applies similar to \cite{bansal2009speedconf}.

Now we bound the $d\Phi_i^{\alg}/dt$ for $i\ge 2$, which is easier, since there is no $\opt$ component in them. Clearly,  $d\Phi_i^{\alg}/dt=0$ when server $i$ is inactive. So next we consider when server $i$ is active.
 Let the size of the job being processed by server $i$ with the algorithm be $q^i$. Then for the algorithm, $n^i(q)$
  decreases by $1$ for $q \in [q^i - s_i dt, q^i]$, and the
  contribution in $d\Phi_i^{\alg}/dt$ because of the algorithm is
  \begin{align}\nn
   d\Phi_i^{\alg}/dt & \le  - c_i \Delta \left( \frac{ \sum_{k=2}^{i-1} n^k + n^i(q^i)}{A+1}\right)s_i, \\ \label{eq:dummy100}
   &=- c_i \Delta \left( \frac{ \sum_{k=2}^{i-1} n^k + n^i}{A+1}\right),
\end{align}
where $\Delta$ has been defined after \eqref{defn:phi}.
Now, we apply technical Lemma \ref{lem:bansal}, to bound RHS of
\eqref{eq:dummy100}. Setting $\beta_\alg = \frac{\sum_{k=2}^{i}
  n^k }{A+1},$ using Lemma \ref{lem:bansal}, note that
\begin{align}\nn
\Delta \left(\beta\right)(-s^i)& \le  \left(-s^i +P^{-1}(\beta_\alg)\right)P^{'}(P^{-1}(\beta_\alg)) -\beta_\alg, \\ 
& \stackrel{(a)}\le   -\beta_\alg, \label{eq:dummy457}
\end{align}
where $(a)$ follows from the speed definition \eqref{eq:speeddef}
which ensures that $P(s^i) \geq \beta_\alg$ (even if $n^1=0$ since $\sum_{i=2}^Kn^i\le A$). Thus, using the Definition \eqref{defn:r} of $r_i$, we get 
\begin{align*}
   d\Phi_i^{\alg}/dt & \le  - c_i \frac{r_i}{A+1}\end{align*}
which completes the proof.
\end{proof}

\begin{remark}
  \label{rem:boundedspeed}
  If the set of allowable speeds is bounded, i.e., $[0,\sfB],$ where
  $P(\sfB) > 1,$ the statement of Theorem~\ref{thm:main} holds as is. The
  main points of difference are in the proof of
  Lemma~\ref{lem:bansalphi1}, in the two applications of
  Lemma~\ref{lem:bansal}. The first application (see
  \eqref{eq:dummy456}) holds due to the justification in
  \cite{bansal2009speedconf}. Since $\beta_\alg\le 1$, the second application (see
  \eqref{eq:dummy457}) requires $P(s^i) \geq \beta_\alg,$ which holds
  so long as $P(\sfB) > 1.$
\end{remark}

\section{Proof of Theorem~\ref{thm:stochastic:constant_comp}}
\label{app:stochastic}

Let $C_i = \Exp{T_i} + \Exp{E_i}$ denote the cost associated with
layer~$i.$ We prove the Theorem via deriving two lower bounds as follows.

\begin{lemma}
  \label{lemma:lb1}
  For any policy, 
  $$\lambda C_i \geq c_i \frac{\rho_i^{\alpha_i}}{m_i^{\alpha_i - 1}}.$$
\end{lemma}

\begin{proof}
  This lower bound is obtained using considering only the energy
  cost. Let $S_{i,j}$ denote the steady state speed of server~$j$ in
  layer~$i.$
\begin{align*}
  \lambda C_i &\geq \lambda \Exp{E_i}
              =\sum_{j = 1}^{m_i} \Exp{P(S_{i,j})}
              \stackrel{(a)}\geq \sum_{j = 1}^{m_i} P(\Exp{S_{i,j}})
              \stackrel{(b)}\geq \sum_{j = 1}^{m_i} P(\rho_i/m_i)
              =c_i \frac{\rho_i^{\alpha_i}}{m_i^{\alpha_i - 1}}.
\end{align*}
Here, $(a)$ follows by applying Jensen's inequality. $(b)$ is a
consequence of the convexity of $P(\cdot)$ along with
$\sum_{j=1}^{m_i} \Exp{S_{i,j}} = \rho_i.$
\end{proof}

\begin{lemma}
  \label{lemma:lb2}
  For any policy, 
  $$\lambda C_i \geq c_i^{1/\alpha_i} \rho_i \alpha_i (\alpha_i - 1)^{1/\alpha_i - 1}.$$
\end{lemma}
\begin{proof}
  This lower bound comes from optimizing the cost of serving a single
  job in isolation. Indeed, 
  \begin{align*}
    \lambda C_i &\geq \min_{s > 0} \frac{\rho_i}{s} + \frac{\rho_i P(s)}{s}.
  \end{align*}
  The first term above is ($\lambda$ times) the delay cost, and the
  second is ($\lambda$ times) the energy cost. The above optimization
  can be solved explicitly, yielding the statement of the lemma.
\end{proof}

We are now ready to prove
Theorem~\ref{thm:stochastic:constant_comp}. Let $C_i^A$ denote the
cost associated with the proposed algorithm. This is given by
\begin{align}\nn
  \lambda C_i^A &= \frac{\rho_i}{s_i^A - \rho_i/m_i} + \frac{\rho_i}{s_i^A} P_i(s_i^A), \\ \label{eq:dummy800}
  &= \rho_i + c_i \rho_i \left(1+\frac{\rho_i}{m_i} \right)^{\alpha_i-1}. 
\end{align}
The above expressions follow since layer~$i$ receives arrivals as per
a Poisson process (this is a consequence of Burke's
theorem~\cite{Harchol2013}), which is further split into $m_i$
independent Poisson streams feeding into the $m_i$ servers in
layer~$i.$ Indeed, the steady state mean response time in layer~$i$
equals $\frac{1}{\mu_i(s_i^A - \rho_i/m_i)},$ and $\lambda$ times the
energy per job equals the stationary power consumption, given by the
second term in \eqref{eq:dummy800}.

Following \eqref{eq:dummy800} and Lemma \ref{lemma:lb1} and Lemma
\ref{lemma:lb2}, we get
\begin{align*}
  \frac{C_i^A}{C_i^*} & \leq \frac{\rho_i + c_i \rho_i \left(1+\frac{\rho_i}{m_i} \right)^{\alpha_i-1}}
                      {\max(c_i^{1/\alpha_i} \rho_i \alpha_i (\alpha_i - 1)^{1/\alpha_i - 1}
                      ,c_i \frac{\rho_i^{\alpha_i}}{m_i^{\alpha_i - 1}})}, \\
  & \leq \frac{\rho_i + c_i \rho_i \left(1+\frac{\rho_i}{m_i} \right)^{\alpha_i-1}}
    {\min \bigl(c_i^{1/\alpha_i} \alpha_i (\alpha_i - 1)^{1/\alpha_i - 1},c_i \bigr)
    \max\bigl(\rho_i,\frac{\rho_i^{\alpha_i}}{m_i^{\alpha_i - 1}}\bigr)}, \\
  &\leq \frac{1 + c_i 2^{\alpha_i-1}}{\min \bigl(c_i^{1/\alpha_i} \alpha_i (\alpha_i - 1)^{1/\alpha_i - 1},c_i \bigr)} =: c_i.
\end{align*}
Finally, we can bound the overall cost of the proposed algorithm as
follows.
\begin{align*}
  C^A = \sum_{i = 1}^K C_i^A \leq \sum_{i = 1}^K c_i C_i* \leq \left(\max_{1 \leq i \leq K} c_i\right) C^*.   
\end{align*}

\bibliography{../refs,refsextra}


\end{document}